\documentclass[epj]{svjour}
\usepackage{graphics}
\begin{document}
\title{Coulomb fields as  perturbations  of type D vacuum space-times}
\author{Koray D\"{u}zta\c{s} 
}                     
\offprints{koray.duztas@emu.edu.tr}          
\institute{Physics Department, Eastern Mediterranean  University, Famagusta, Northern Cyprus, Mersin 10, Turkey 
}
\date{Received: date / Revised version: date}
%
\abstract{
We consider a test, non-null electromagnetic field special in the sense that the principal null directions of the field lie along  the two repeated principal null directions of the type D vacuum background. We prove that the special non-null field is a valid solution of the Maxwell equations on a type D background. This field behaves as $1/r^2$ and the only non-vanishing Maxwell scalar is $\phi_1$ so it can be called the Coulomb field. We show that the contribution to the energy and angular momentum fluxes of the Coulomb field vanishes at any surface of integration, on the type D vacuum backgrounds  Kerr, Kerr-Taub-NUT, Taub-NUT, and Schwarzschild. Therefore this field does not lead to any perturbations of the mass and angular momentum parameters of these space-times.
\PACS{
      {04.20.Gz }{Spacetime topology, causal structure, spinor structure}   \and
      {04.70.Bw	}{Classical black holes}
     } 
} 
\maketitle
\section{Introduction}
A space-time can be identified conveniently by its response to external perturbations. In that context, the study of black hole perturbations has gained considerable interest both in relativistic astrophysics and mathematical physics. The first achievement in this area regarding the electromagnetic field was to reduce Maxwell's equations in Kerr geometry to a single second order partial diferential equation \cite{ipser}. The most important contribution is definitely due to Teukolsky and Press who decoupled  and separated wave equations and derived asymptotic solutions for electromagnetic gravitational and neutrino perturbations of Kerr black holes, also analysed the interaction of black holes with test fields~\cite{teuk1,teuk3}. Teukolsky  decoupled the wave equations for the perturbations in a general context which applies to all type D vacuum space-times. This is achieved by employing the Newman-Penrose two spinor formalism~\cite{newpen}. Then he proceeds to apply separation of variables on Kerr background. As Kerr geometry admits two Killing vectors $K_1^a=(\partial/\partial t)^a$ and $K_2^a=(\partial/\partial \phi)^a$ the separation of the $t$ and  $\phi$ coordinates is expected. It turns out that the $r$ and $\theta$ dependence of the perturbations can also be separated. A new version of Teukolsky's master equation was derived in~\cite{teukbini}. The decoupling and separation of variables of the perturbation equations was studied on various type D vacuum backgrounds including the Taub-NUT~\cite{bini0},  the Kerr-Taub-NUT~\cite{bini}, and the C metric~\cite{kofron}.
In particular, there exists detailed analysis of the electromagnetic perturbations of space-times~\cite{cohen1,cohen2,chrza1,chrza2,chandra1,chandra2,fayos1}.

In this work we consider a test, non-null electromagnetic field in type D vacuum space-times which is special in the sense that the principal null directions of the field lie along the two principal null directions of the type D vacuum background. For such special fields the decoupling and separation of variables of Maxwell's equations is trivial. The problem is to prove whether this field can exist on a type D vacuum background; i.e. it is a valid solution of Maxwell equations on this background.

The outline of the papers is as follows: In the introduction we briefly review the necessary formalisms. In section (\ref{sec}) we prove the existence of the special electromagnetic field and evaluate its behaviour on type D vacuum backgrounds. In section (\ref{sec3}) we calculate the energy and angular momentum fluxes due to the special field  to calculate the perturbations of the mass and angular momentum parameters of the background geometry.

\subsection{Newman-Penrose formalism} In most of the relevant works Newman Penrose (NP) two spinor formalism is used~\cite{newpen}. This formalism is based on a spin basis 
 $(o, \iota)$  endowed with a symplectic structure $\epsilon_{AB}=-\epsilon_{BA}$.
\begin{equation}
\begin{array}{l}
\epsilon_{AB} o^A o^B=\epsilon_{AB}\iota^A \iota^B=0 \\
\epsilon_{AB}o^A \iota^B=1
\end{array}
\end{equation}
The spin basis induces a tetrad of null vectors.
\begin{equation}
l^a=o^A \bar{o}^{A'} \quad n^a=\iota^A \bar{\iota}^{A'} \quad m^a=o^A \bar{\iota}^{A'} \quad \bar{m}^a=\iota^A \bar{o}^{A'} \label{nptetrad}
\end{equation}
The NP null tetrad satisfies
\begin{eqnarray}
& &l_an^a=n_al^a=-m_a \bar{m}^a=-\bar{m}_a m^a=1 \nonumber \\
& &l_am^a=l_a \bar{m}^a=n_am^a=n_a\bar{m}^a=0
\label{ortog}
\end{eqnarray}
The projection of derivatives onto the null tetrad are defined by conventional symbols.
\begin{eqnarray}
&=&D=l^a\nabla_a,\quad \Delta=n^a\nabla_a,\quad \delta=m^a\nabla_a,\quad \bar{\delta}=\bar{m}^a\nabla_a \nonumber \\
&=& n_aD + l_a \Delta - \bar{m}_a \delta -m_a \bar{\delta} \label{nablanp}
\end{eqnarray}
NP spin coefficients are defined by
\begin{equation}
\begin{array}{ll}
Do_A=\epsilon o_A - \kappa \iota_A & D\iota_A=\pi o_A-\epsilon \iota_A \\
\Delta o_A=\gamma o_A -\tau \iota_A & \Delta \iota_A=\nu o_A -\gamma\iota_A \\
\delta o_A=\beta o_A-\sigma \iota_A & \delta \iota_A =\mu o_A -\beta \iota_A \\
\bar{\delta}o_A =\alpha o_A-\rho \iota_A & \bar{\delta} \iota_A=\lambda o_A-\alpha \iota_A
\end{array} \label{npscalars}
\end{equation}
One can also define the spin coefficients in the form $o^A D o_A=\kappa, o^A D \iota_A=\epsilon $ , which is manifest in (\ref{npscalars}).

Since we are only dealing with scalars for any NP quantitiy $\phi$, $\nabla_{[a}\nabla_{b]} \phi=0$. This leads to commutation relation for NP derivative operators.
\begin{eqnarray}
(\delta D-D\delta)\psi &=&[(\bar{\alpha}+\beta - \bar{\pi})D+\kappa \Delta - (\bar{\rho}+\epsilon-\bar{\epsilon} )\delta -\sigma \bar{\delta}]\psi \nonumber \\
(\Delta D- D\Delta)\psi &=& [(\gamma +\bar{\gamma} )D+ (\epsilon +\bar{\epsilon})\Delta - (\bar{\tau} +\pi)\delta -(\tau + \bar{\pi})\bar{\delta}]\psi \nonumber \\
(\delta \Delta -\Delta\delta)\psi &=& [-\bar{\nu} D+(\tau -\bar{\alpha} -\beta)\Delta +(\mu -\gamma + \bar{\gamma})\delta + \bar{\lambda}\bar{\delta}]\psi \nonumber \\
(\bar{\delta}\delta -\delta \bar{\delta})\psi &=& [(\bar{\mu}-\mu )D + (\bar{\rho} -\rho)\Delta +(\alpha - \bar{\beta})\delta - (\bar{\alpha} - \beta)\bar{\delta}]\psi  \label{npcommuters}
\end{eqnarray} 
The commutation relations (\ref{npcommuters}) are identically satisfied for the derivative operators acting on any scalar $\psi$. 
\subsection{Electromagnetic and gravitational perturbations in NP formalism}
Let us first introduce the totally antisymmetric Maxwell tensor $F_{ab}$ with its spinor equivalent $F_{ABA'B'}$. We define
\begin{equation}
\phi_{AB} \equiv \frac{1}{2}F_{ABC'} {}^{C'} =\phi_{BA}
\end{equation}
This leads to
\begin{equation}
F_{ABA'B'}=\phi_{AB} \epsilon_{A'B'}+\epsilon_{AB} \bar{\phi}_{A'B'}
\end{equation}
The form of source free Maxwell equations in NP formalism is very simple.
\begin{equation}
\nabla^{AA'}\phi_{AB}=0 \label{maxwellnp}
\end{equation}
There are four complex equations here by $A',B=0,1$, corresponding to eight real Maxwell equations. The symmetric valence 2 spinor $\phi_{AB}$ generates 3 complex scalars via
\begin{equation}
\phi_0 =\phi_{AB}o^A o^B \quad \phi_1 =\phi_{AB}o^A \iota^B \quad \phi_2 =\phi_{AB}\iota^A \iota^B \label{npscalarsmax}
\end{equation}
The NP description of electromagnetism is given in terms of the scalars (\ref{npscalarsmax}). Also note that
\begin{equation}
\phi_{AB}=\phi_2 o_A o_B -2\phi_1 o_{(A }\iota_{B)}
+\phi_0 \iota_A \iota_B \label{phiab}
\end{equation}
The explicit form of Maxwell's equations in terms of NP scalars can be derived.
\begin{equation}
\begin{array}{l}
(D-2\rho)\phi_1 -(\bar{\delta} +\pi-2\alpha)\phi_0 +\kappa \phi_2=0 \\
(D-\rho +2\epsilon)\phi_2 -(\bar{\delta} +2\pi)\phi_1 +\lambda \phi_0 =0 \\
(\delta -2\tau)\phi_1 -(\Delta +\mu -2\gamma)\phi_0 -\sigma \phi_2 =0 \\
(\delta -\tau +2\beta)\phi_2 -(\Delta +2\mu)\phi_1 -\nu \phi_0 =0 
\end{array}
\label{maxwell}
\end{equation}

Any totally symmetric spinor can be decomposed  in terms of univalent spinors (see e.g. \cite{agr,penrosebook}). Hence we may decompose the spinor equivalent of Maxwell tensor.
\begin{equation}
\phi_{AB}=\alpha_{(A}\beta_{B)}
\end{equation}
$\alpha$ and $\beta$ are called the principal spinors of $\phi_{AB}$. If $\alpha$ and $\beta$ are proportional then $\alpha$ is called a repeated  principal  spinor of $\phi$, and $\phi$ is said to be algebraically special or null, or of type $N$. The corresponding real null vector $\alpha_a=\alpha_A \bar{\alpha}_{A'}$ is called a repeated principal null direction.
If $\alpha$ and $\beta$ are not proportional then $\phi$ is said to be algebraically general or type I, or non-null.

To formulate gravity let us define the spinor equivalent of the Weyl tensor $C_{abcd}$. 
\begin{equation}
C_{abcd}+ i C^*_{abcd}=2\Psi_{ABCD}\epsilon_{A'B'}\epsilon_{C'D'}
\end{equation} 
$\Psi_{ABCD}$ is totally symmetric and satisfies the spinor analogue of Bianchi identities in vacuum 
\begin{equation}
\nabla^{DD'}\Psi_{ABCD}=0 \label{bianchi}
\end{equation}
The explicit form of Bianchi identities are derived in \cite{newpen}. (Also see \cite{agr,penrosebook}) Since $\Psi_{ABCD}$ is totally symmetric there exists univalent spinors  $\alpha_A , \beta_B, \gamma_C, \delta_D$ such that 
\begin{equation}
\Psi_{ABCD}=\psi \alpha_{(A}  \beta_B \gamma_{C} \delta_{D)} \label{theoem1}
\end{equation}
$\alpha_A , \beta_B, \gamma_C, \delta_D$ are called the principal spinors of $\Psi_{ABCD}$.  The corresponding real null vectors determine the principal null directions of $\Psi_{ABCD}$. The classification of space-times according to the principal null directions of the Weyl tensor is known as Petrov classification. As in the case of electromagnetism if none of the principal null directions coincide, the space-time is algebraically general of type I, and if all four principal null directions coincide the space-time is of type N. 
If there are two pairs of repeated principal null directions the space-time is of type D. 
\subsection{The conditions for the existence of special fields}
Algebraically special  (null) electromagnetic fields are known to exist in space-times that admit a shear free, geodesic null congruence~\cite{robinson}. 
Existence of the fields refers to the fact that the integrability conditions are satisfied by Maxwell equations  (also see e.g. \cite{penrosebook,agr}). Let us consider a null electromagnetic field with principal direction $l^a$ such  that the only non-vanishing Maxwell scalar is $\phi_2=\phi$. Maxwell equations (\ref{maxwell}) reduce to
\begin{equation}
D\phi=\rho \phi, \quad \delta \phi =(\tau - 2\beta) \phi \label{nullmax1}
\end{equation}
A solution for (\ref{nullmax1}) does not necessarily exist, since the commutation relations (\ref{npcommuters}) should also be satisfied. Using NP field equations one derives that a solution for (\ref{nullmax1}) exists if the spin coefficients satisfy $\kappa=\sigma=0$, i.e. if the background spacetime admits a geodesic, shear free null congruence with tangent vector $l^a$.  (see \cite{penrosebook,agr})  

\section{The special non-null electromagnetic field}\label{sec}
It is known that electrovacuum solutions for  type D space-times exist such that the two repeated principal null congruences of the Weyl
tensor are aligned with the two principal null congruences of the non-null electromagnetic
field~\cite{plebanski}. (Also see \cite{podolsky} and references therein). In this work, we consider an algebraically general (non-null) test electromagnetic field in type D vacuum space-times, which is special in the sense that the principal null directions of the electromagnetic field lie along  the repeated principal null directions of the space-time. First, we have to prove that such a special field is a valid solution of Maxwell's equations on type D vacuum background.
\begin{theorem}\label{theo}
Type D  vacuum space-times admit a special, test, non-null electromagnetic field,  such that the two principal null directions of the electromagnetic field  lie along the repeated principal null directions of the space-time.
\end{theorem}
\begin{proof}
Naturally we choose a spinor basis $(o,\iota)$ for a type D space-time such that the two principal null directions correspond to $l^A=o^Ao^{A'}$ and $n^A=\iota^A \iota^{A'}$. Consider an algebraically general test electromagnetic field
\begin{equation}
\phi_{AB}=\phi o_{(A}\iota_{B)} \label{phispecial}
\end{equation}
which is special in the sense that its principal null directions are parallel to those of the background space-time. (A test field is one that has a negligible effect on the background geometry.) From (\ref{npscalarsmax}) and (\ref{phiab}) we see that  $\phi_0=\phi_2=0$.  The only non-vanishing Maxwell scalar is $\phi_1$, which equals $-\phi/2$ according to our definition (\ref{phispecial}). Then, Maxwell equations have the form:
\begin{equation}
\begin{array}{l}
D \phi_1 =2\rho \phi_1 \\
\Delta \phi_1 =-2\mu \phi_1 \\
\delta \phi_1 =2\tau \phi_1 \\
\bar{\delta} \phi_1 =-2\pi \phi_1 \\
\end{array}
\label{maxspecial}
\end{equation} 
We have to prove that a solution for the system (\ref{maxspecial}), i.e. a solution  of Maxwell's equation such that the only non-vanishing Maxwell scalar is $\phi_1$, exists in a type D background. If a space-time is of type D the only non-vanishing scalar of the Weyl tensor is $\Psi_2$. In this case  the  Bianchi identities in vacuum reduce to
\begin{equation}
\begin{array}{l}
D \Psi_2 =3\rho \Psi_2 \\
\Delta \Psi_2 =-3\mu \Psi_2 \\
\delta \Psi_2 =3\tau \Psi_2 \\
\bar{\delta} \Psi_2 =-3\pi \Psi_2\\
\end{array}
\label{bianchispe}
\end{equation} 
The integrability conditions for the systems (\ref{maxspecial}) and (\ref{bianchispe}) are identical. In other words the integrability conditions for the existence of a special test electromagnetic field in the form (\ref{phispecial}), in a type D background, are identical with the conditions for the existence of the type D background itself. Thus, there exists a special non-null test electromagnetic field in every type D background.
\end{proof}
The integrability conditions for the existence of type  D vacuum space-times were derived by Kinnersley~\cite{kinner}. Let us re-derive these conditions in a way that clarifies their equivalence with the integrability conditions for the system (\ref{maxspecial}), as suggested by theorem (\ref{theo}).   Let $\theta=\ln \phi_1$ for (\ref{maxspecial}) and $\theta=\ln \Psi_2$  for (\ref{bianchispe}). Then the systems reduce to $(\kappa=\sigma=\lambda=\nu=\epsilon=0)$
\begin{equation}
D \theta =C\rho, \quad
\Delta \theta=-C\mu, \quad 
\delta \theta =C\tau, \quad
\bar{\delta} \theta =-C\pi 
\end{equation}
where $C=2$ for (\ref{maxspecial}) and $C=3$ for (\ref{bianchispe}). First consider the commutation relation for $D$ and $\Delta$.
\begin{eqnarray}
(\Delta D- D\Delta)\theta &=& [(\gamma +\bar{\gamma} )D  - (\bar{\tau} +\pi)\delta -(\tau + \bar{\pi})\bar{\delta}]\theta  \nonumber \\
&=& C [\rho(\gamma  + \bar{\gamma}) - \tau (\bar{\tau} + \pi) + \pi (\tau + \bar{\pi})] \nonumber \\
&=&  C [\rho(\gamma  + \bar{\gamma}) - \tau \bar{\tau} + \pi \bar{\pi}]
\label{condi01}
\end{eqnarray}
We also have
\begin{equation}
(\Delta D- D\Delta)\theta = \Delta(C\rho) - D(-C\mu)=C(\Delta\rho + D\mu)
\label{condi02}
\end{equation}
The first integrability condition is derived by equating (\ref{condi01}) and (\ref{condi02})
\begin{equation}
\Delta\rho + D\mu=\rho(\gamma  + \bar{\gamma}) - \tau \bar{\tau} + \pi \bar{\pi}
\label{condi1}
\end{equation}
The condition (\ref{condi1}) is not identically satisfied by NP field equations. Similarly the commutators $(\bar{\delta} D- D \bar{\delta})$ and $(\bar{\delta}\Delta - \Delta \bar{\delta})$ lead to non-trivial conditions that are not identically satisfied by NP field equations. (The other commutators lead to conditions that are identically satisfied by NP field equations as in the case of null electromagnetic fields.)  
\begin{eqnarray}
& & \bar{\delta}\rho + D \pi = \rho (\alpha + \bar{\beta}) \label{condi2} \\
&& \delta \mu + \Delta \tau = \tau (\gamma - \bar{\gamma}) - \mu (\bar{\alpha} + \beta ) \label{condi3}
\end{eqnarray}
The conditions (\ref{condi1}), (\ref{condi2}), and (\ref{condi3}) are imposed on the background space-time so that a solution for the system (\ref{maxspecial}) exists (the case $C=2$). Also they are the conditions that a type D space-time itself exists (the case $C=3$).
\subsection{Behaviour  in Kerr space-time}
We have proved the existence of the special non-null electromagnetic field in type D vacuum space-times. For that field the only non-vanishing Maxwell scalar is $\phi_1$. We are particularly interested in its asymptotic behaviour in Kerr space-time. From NP field equations for type D space-times, we have $D\rho=\rho^2$. This leads to the solution 
\begin{equation}
\phi_1=\rho^2 C_1 \label{solution}
\end{equation}
where $C_1$ is independent of $r$. In Kerr space-time $\rho=-(r- i a \cos \theta)^{-1}$ so the special field behaves as $1/r^2$ everywhere. In fact, the general expression for $\rho$ in a type D space-time is  $\rho=-(r+ i \rho^0)^{-1}$, where $\rho^0$ is independent of $r$~\cite{kinner}, so $1/r^2$ behaviour applies to every type D vacuum space-time. Since the only non-vanishing Maxwell scalar for the special non-null electromagnetic field is $\phi_1$, and the field behaves as $1/r^2$, it can be called the Coulomb field around the black holes.
\section{Perturbations of the space-time parameters due to Coulomb fields}\label{sec3}
The effect of a test field on the background geometry is negligible, by definition. Though the interaction of a test field with the background space-time is not supposed to change the  structure of the metric tensor, it is expected to lead to some perturbations on the  parameters of the background metric. In this section we evaluate the perturbations of the mass ($M$), and angular momentum ($a$) parameters  of Kerr and Kerr-Taub-NUT metrics, due to the Coulomb field. The existence of the Coulomb field on these backgrounds is guaranteed  by theorem (\ref{theo}), since they are type D vacuum solutions of Einstein's equations.

Kerr solutions are stationary and axisymmetric so that they admit two Killing vectors $\partial /\partial t$ and $\partial / \partial \phi$. This property allows the definition of globally conserved energy and angular momentum for the test fields or particles in this space-time. In this respect, the rates of change in the corresponding background parameters can be expressed as fluxes through a space-like hyper-surface. The current conservation equation $\nabla_a (T^{ac}K_c)=0$ (where $K$ is a Killing vector) leads to
\begin{equation}
\left(\frac{dM}{dt}\right) = - \int_{S} \sqrt{-g} \, T^{1}_{\;\;0} d\theta d\phi
\label{eq:dm/dt}
\end{equation}
and
\begin{equation}
\left(\frac{dL}{dt}\right)
=   \int_{S} \sqrt{-g} \, T^1_{\;\;3} d\theta d\phi \label{eq:dl/dt}
\end{equation}
where $S$ is the surface of integration. In equations (\ref{eq:dm/dt}) and (\ref{eq:dl/dt}), the metric signature $(+,-,-,-)$ is adopted. The energy momentum tensor for electromagnetic fields in terms of Maxwell's scalars, is given by
\begin{eqnarray}
4\pi T_{\mu\nu}&=&\{\phi_0 \phi_0^* n_\mu n_\nu + 2\phi_1 \phi_1^* [l_{(\mu}n_{\nu)}+ m_{(\mu}m_{\nu)}^*] +\phi_2 \phi_2^*  l_{\mu}l_{\nu}\nonumber \\
& &-4\phi_1 \phi_0^* n_{(\mu}m_{\nu)}-4\phi_2 \phi_1^* [l_{(\mu}m_{\nu)}+2\phi_2 \phi_0^* m_{\mu}m_{\nu} \} +\mbox{c.c.}
\label{tmunu}
\end{eqnarray} 
where c.c. denotes complex conjugate. For the Coulomb  field, the only non-vanishing Maxwell scalar is $\phi_1$.The energy momentum tensor reduces to
\begin{equation}
4\pi T_{\mu\nu}=  2\phi_1 \phi_1^* [l_{(\mu}n_{\nu)}+ m_{(\mu}m_{\nu)}^*] +\mbox{c.c.}\label{tmunured}
\end{equation}
Now consider the  NP tetrad  for the  Kerr-Taub-NUT  metric. 
\begin{eqnarray}
& &l^\mu=[(\Sigma + aA)/\Delta ,1,0,a/\Delta],\quad n^\mu=[(\Sigma + aA), -\Delta ,0,a]/(2\Sigma) \nonumber \\
& &m^\mu=[A \csc \theta ,0,-i, \csc \theta]/[\sqrt{2}(\ell -ir+a \cos \theta)]
\label{tetrad0}
\end{eqnarray}
where $\Sigma=r^2+(\ell + a^2\cos\theta)^2$, $\Delta=r^2 - 2Mr-\ell^2 +a^2$, and $A=a \sin^2 \theta - 2 \ell \cos \theta$. Tetrad (\ref{tetrad0}) represents the Kerr metric when $\ell=0$, the Taub-NUT metric when $a=0$, and the Schwarzschild metric when $\ell=a=0$. Using  this tetrad, one can derive
\begin{equation}
[l_{(1}n_{0)}+ m_{(1}m_{0)}^*]=[l_{(1}n_{3)}+ m_{(1}m_{3)}^*]=0 \label{tmunuzero} 
\end{equation}
(\ref{tmunured}),(\ref{tetrad0}) and (\ref{tmunuzero}) imply that $T^1_{\;\; 0}$  and $T^1_{\;\; 3}$ identically vanish.  Substituting these in (\ref{eq:dm/dt}) and (\ref{eq:dl/dt}), one derives 
\begin{equation}
\left(\frac{dM}{dt}\right)_{\rm b.h} =\left(\frac{dL}{dt}\right)_{\rm b.h} =0
\end{equation}
The contribution of the Coulomb field  to the radial energy and angular momentum fluxes  vanishes  at every surface of integration. Thus, the perturbations for the mass and angular momentum parameters due to the Coulomb field vanishes, for the type D vacuum backgrounds Kerr, Kerr-Taub-NUT, Taub-NUT, and Schwarzschild. The space-times identically retain their initial states since the geometry is uniquely determined by the parameters $(M,a,\ell)$. 
\section{Summary and conclusions}
In this work, we first proved that  type D  vacuum space-times admit a special test electromagnetic field such that the principal null directions of the electromagnetic field coincide with those of the space-time. This corresponds to a solution of Maxwell equations such that the only non-vanishing Maxwell scalar is $\phi_1$. It turns out that the integrability conditions for the existence of the relevant electromagnetic field in a type D  vacuum space-time are identical with the conditions of the existence of the  space-time itself. Next, we stated that the solution for this electromagnetic field has the form $\phi_1=\rho^2 C_1$, so it behaves as $1/r^2$ everywhere, and in every type D vacuum space-time. This special field can be called the Coulomb field.  We showed that  the contribution of the Coulomb field to energy and angular momentum fluxes vanishes at every surface of integration, for the Kerr, Kerr-Taub-NUT, Taub-NUT, and Schwarzschild backgrounds. Thus, the Coulomb field does not lead to any perturbations of the mass and angular momentum parameters of these space-times.


%
%
%

%
%

\end{document}